\documentclass[cmp]{svjour}
\usepackage{amssymb, amsmath,latexsym}

\usepackage[dvips]{graphicx}

\journalname{Communications in Mathematical Physics}

\newtheorem{thm}    {Theorem}
\newtheorem{lem}     {Lemma}
\newtheorem{cor}  {Corollary}
\newtheorem{prop}        {Proposition}

\def\Tr{\mathop{\rm Tr}\nolimits}

\def\rank{\mathop{\rm rank}\nolimits}

\def\MSE{\mathop{\rm MSE}\nolimits}

\def\integer{\mathbb{Z}}

\def\R{\mathbb{R}}
\def\C{\mathbb{C}}

\def\id{\mathop{\rm id}\nolimits}

\def\rank{\mathop{\rm rank}\nolimits}

\def\complex{\mathbb{C}}

\def\Label#1{\label{#1}\ [\ #1\ ]\ }
\def\Label{\label}

\begin{document}

\title{{Comparison between the
Cramer-Rao and the mini-max approaches in quantum channel estimation}}
\titlerunning{Comparison between Cramer-Rao and mini-max approaches}

\author{Masahito Hayashi}
\institute{Graduate School of Information Sciences, Tohoku University, Sendai, 980-8579, Japan.\\ 
Centre for Quantum Technologies, National University of Singapore, 3 Science Drive 2, Singapore 117542\\
\email{hayashi@math.is.tohoku.ac.jp}}
\authorrunning{Masahito Hayashi}

\date{Received:}
\communicated{name}

\maketitle
\begin{abstract}
In a unified viewpoint
in quantum channel estimation,
we compare
the Cram\'{e}r-Rao and the mini-max approaches,
which gives the Bayesian bound in the group covariant model.
For this purpose, we introduce the {\it local asymptotic mini-max} bound, 
whose maximum is shown to be equal to the asymptotic limit of the mini-max bound.
It is shown that the local asymptotic mini-max bound is strictly larger than the Cram\'{e}r-Rao bound
in the phase estimation case
while the both bounds coincide when 
the minimum mean square error decreases with the order $O(\frac{1}{n})$.
We also derive a sufficient condition for that 
the minimum mean square error decreases with the order $O(\frac{1}{n})$.
\end{abstract}

\section{Introduction}\Label{s1}

In quantum information technology,
it is usual to use quantum channel for sending quantum state.
Since a quantum channel has noise, it is important to identify quantum channel.
In this paper, we consider 
theoretical optimal performance of quantum channel estimation
when we can apply the same unknown channel several times.
In order to treat this problem,
we employ quantum state estimation theory.
In our setting,
we can optimize our input state and our measurement\cite{F4,FI2,FI1,F3,IF,Ha1,CDPS,BB,Luis,BDR,IH,K,HKO,keiji}.
As is illustrated in Fig. \ref{f1} with $n=4$,
it is assumed to be possible 
to use entanglement with reference system in the measurement process
when 
the channel $\Lambda_\theta$ with the unknown parameter $\theta$ is applied $n$ times.
This setting is mathematically equivalent with the setting given in Fig. \ref{f2} with $n=4$,
which has a single input state in the large input system and a single measurement in the large output system.
In this paper, we consider quantum channel estimation with
the formulation given by Fig. \ref{f2}.

\begin{figure}[htbp]
\begin{center}
\scalebox{1.0}{\includegraphics[scale=0.4]{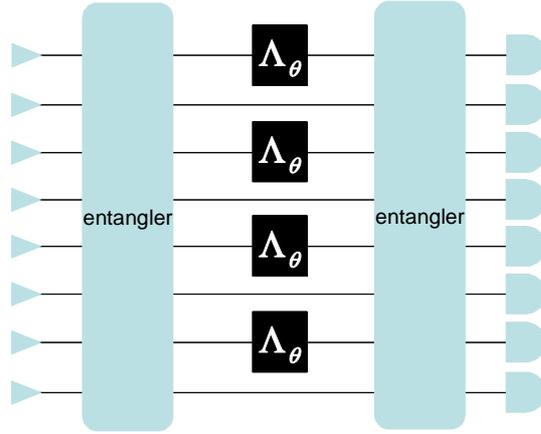}}
\end{center}
\caption{Estimation scheme of quantum channel}
\Label{f1}
\end{figure}%

\begin{figure}[htbp]
\begin{center}
\scalebox{1.0}{\includegraphics[scale=0.4]{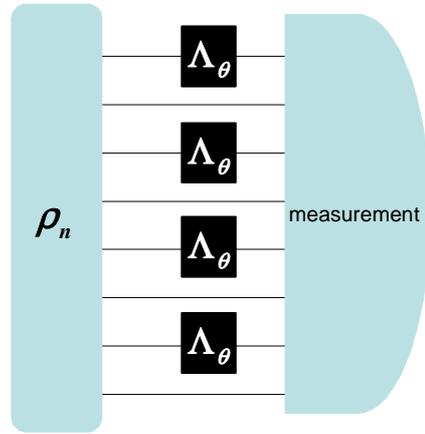}}
\end{center}
\caption{Simpler estimation scheme of quantum channel}
\Label{f2}
\end{figure}%

In the state estimation,
when the number $n$ of prepared states goes to infinity,
the mean square error (MSE) behaves as the order $O(\frac{1}{n})$
as in the estimation of probability distribution.
However, in the estimation of quantum channel, two different analyses
were reported concerning asymptotic behavior of MSE.
As the first case,
in the estimations of depolarizing channels and Pauli channels,
the optimal MSE behaves as $O(\frac{1}{n})$\cite{F4,FI1}.
As the second case,
in the estimation of unitary,
the optimal MSE behaves as $O(\frac{1}{n^2})$\cite{F3,IF,Ha1,CDPS,BB,Luis,BDR,IH}.
In the second case, two different types of results were reported:
One is based on the Cram\'{e}r-Rao approach\cite{F3,IF}.
The other is based on the mini-max approach\cite{Ha1,CDPS,BB,Luis,BDR,IH}.

The Cram\'{e}r-Rao approach is based on the notion of locally unbiased estimator, and allows one
to give a simple lower bound (the Cram\'{e}r-Rao bound) to the mean square error
(MSE) at a given point. 
The mini-max approach aims to minimize the
maximum of the MSE over all possible values of the parameter.
The mini-max approach is more meaningful than the Cram\'{e}r-Rao approach
because the true value of the parameter is unknown.
So, the Cram\'{e}r-Rao bound is just a lower bound in general, 
while it can be asymptotically acheived in the case of quantum state estimation 
with the $n$ copies of the unknown state.
However, the Cram\'{e}r-Rao bound has been considered so many times in the literatures\cite{F4,FI1,F3,IF}. 
The reason seems that computing the mini-max bound is a much harder problem. 
Indeed, there are only very few examples of calculations of the mini-max bound, 
and most of them are in the compact group covariant setting,
in which, as was shown by Holevo \cite{HoC},
the mini-max bound coincides with the Bayesian average of the MSE 
over the normalized invariant measure.
So, many researchers \cite{Ha1,CDPS,BB,Luis,BDR,IH} applied this approach
to quantum channel estimation under the group covariance.
As the most simple case for unitary estimation,
phase estimation has been treated with the mini-max approach
by several papers\cite{Luis,BDR,IH}, 
and it has been shown that the minimum MSE that behaves as $O(\frac{1}{n^2})$.
On the other hand, the Cram\'{e}r-Rao approach suggests the noon state as the optimal input\cite{Giovannetti1,Giovannetti2}
in phase estimation.
The later estimation scheme was experimentally demonstrated 
in the case of $n=4$ \cite{take1,take2}and $n=10$ \cite{Mag}.
Also, another group \cite{Heisenberg-limit} experimentally demonstrated an estimation 
protocol concerning the group covariant approach proposed by \cite{Kitaev}.
In phase estimation with group covariant framework,
the asymptotic minimum MSE behaves as $O(\frac{1}{n^2})$\cite{Luis,BDR}.
When we focus on the difference $\hat{\theta}-\theta$ between 
the true parameter $\theta$ and the estimate $\hat{\theta}$,
the limiting distribution concerning the random variable $n(\hat{\theta}-\theta)$
can be obtained through Fourier transform of a function with a finite domain \cite{IH}.
However, the Cram\'{e}r-Rao bound is different from the asymptotic limit of the mini-max bound.
So, these results seem to contradict with each other.
No existing research compares both approaches in a unified viewpoint.

The manuscript consists of two parts. 
In the first part, we discuss
the Cram\'{e}r-Rao approach for channel estimation, 
and give a simple formula to
compute a bound based on the right logarithmic derivative introduced by Holevo\cite{2}.
This formula is based on the Choi-Jamiolkowski representation matrix \cite{Choi,Jami} 
of the quantum channel,
and holds under the condition that, at the given point, 
the support of the Choi-Jamiolkowski representation matrix 
contains the support of its derivative.
Under this condition, we prove the additivity of the RLD Fisher information for quantum channels. 

The second part is about the local asymptotic mini-max bound\cite{Hajek}, 
in which the maximum of the MSE is taken over an interval, 
the width of which is sent to zero after computing the asymptotic limits. 
The obtained results are summarized as follows:
(1) The local asymptotic mini-max bound is strictly larger than the Cram\'{e}r-Rao bound in the case of phase estimation.
(2) Both bounds coincide with each other when 
the asymptotic minimum MSE behaves as $O(\frac{1}{n})$.
(3) The local asymptotic mini-max bound is achievable (Proposition 3),
using a variation of the two-step strategy \cite{5,4}. 
That is, there is an optimal sequence of estimators that achieves the local asymptotic mini-max bound at any point.
As consequence, the maximum of the local asymptotic mini-max bound is shown to be equal to the asymptotic limit of the mini-max bound.

However, the Cram\'{e}r-Rao bound has a different type of achievability.
That is, there is an optimal sequence of estimators that achieves the Cram\'{e}r-Rao at a given point. 
This characteristic is a curious quantum analogue of superefficiency in classical statistics.
In estimation of probability distribution,
if we assume a weaker condition for our estimator,
there exists an estimator that surpasses the Cram\'{e}r-Rao bound only in measure zero points.
Such an estimator is called a superefficient estimator\cite{L}.
In the estimation of phase action, we point out that 
the Cram\'{e}r-Rao bound can be attained only by a quantum channel version of a superefficient estimator that works at specific points.
Since the Cram\'{e}r-Rao approach is based on the asymptotically locally unbiased condition, 
we can conclude that the asymptotically locally unbiased condition is too weak for 
deriving the local asymptotic mini-max bound, which can be attained in all points.
Indeed, a similar phenomenon happens in quantum state estimation 
when we use the large deviation criterion\cite{Ha2}.

This paper is organized as follows.
Some of obtained results are based on quantum state estimation with 
the Cram\'{e}r-Rao approach.
Section \ref{s2} is devoted to a review of the Cram\'{e}r-Rao approach in quantum state estimation.
In this section, the symmetric logarithmic derivative (SLD) Fisher information
and the right logarithmic derivative (RLD) Fisher information
are explained.
However, the Cram\'{e}r-Rao bound 
is obtained only by the locally unbiased condition.
So, it is needed to discuss its relation with the estimator that works globally.
In Section \ref{s3}, we treat
SLD Fisher information and RLD Fisher information
in the quantum channel estimation,
and discuss the increasing order of 
SLD Fisher information.
In Section \ref{s31},
we give several examples where the maximum SLD Fisher information increases with $O(n^2)$.
In Section \ref{s4}, we compare the local asymptotic mini-max bound and the Cram\'{e}r-Rao bound.
We also show the global attainability of the local asymptotic mini-max bound in the channel estimation.
It is also shown that the Cram\'{e}r-Rao bound is closely related to a quantum channel version of superefficiency in the phase estimation.

\section{Cram\'{e}r-Rao bound in quantum state estimation}\Label{s2}
In quantum state estimation,
we estimate the true state through the quantum measurement 
under the assumption that the true state 
of the given quantum system ${\cal H}$ belongs to a certain parametric state family
$\{\rho_\theta | \theta \in \Theta \subset \R^d \}$.
In the following, 
we consider the case when the number $d$ of parameters is one.
Usually, we assume that
$n$ quantum systems are prepared in the state $\rho_\theta$.
Hence, the total system is described by the tensor product space ${\cal H}^{\otimes n}$,
and the state of the total system is given by $\rho_\theta^{\otimes n}$.

In this case, when we choose a suitable measurement, 
the MSE decreases in proportion to $n^{-1}$ as in the estimation of
probability distribution.
So, we focus on the first order coefficient of the 
MSE concerning $\frac{1}{n}$.
In the most general setting, any positive operator valued measure (POVM) $M^n$
on the total system ${\cal H}^{\otimes n}$ 
is allowed as an estimator
when it takes values in the parameter space $\Theta \subset \R$.
The MSE is given as
\begin{align*}
\MSE_\theta(M^n):=
\int (\hat{\theta}-\theta)^2 \Tr \rho_\theta^{\otimes n} M^n(d \hat{\theta}).
\end{align*}

In the quantum case, there are several quantum extensions of Fisher information
when the state $\rho_\theta$ is differentiable at $\theta$
and $(I-P)\frac{d \rho_\theta}{d \theta}(I-P)=0$, where
$P$ is the projection to the support of $\rho_\theta$.
The largest one is the right logarithmic derivative (RLD) Fisher information
$J^R_\theta$, 
and the smallest one is symmetric logarithmic derivative (SLD)
Fisher information $J^S_\theta$.
For these definitions, we define the RLD $L_\theta^R$ and the SLD $L_\theta^S$ 
as the operators satisfying 
\begin{align*}
\frac{d \rho_\theta}{d \theta}=
\rho_\theta L_\theta^R , \quad
\frac{d \rho_\theta}{d \theta}=
\frac{1}{2}
\left(L_\theta^S \rho_\theta+\rho_\theta L_\theta^S \right) .
\end{align*}
Then, the RLD and SLD Fisher informations are given by \cite{1,2,3}
\begin{align*}
J^R_\theta := \Tr \rho_\theta L_\theta^R (L_\theta^R)^\dagger, \quad
J^S_\theta := \Tr \rho_\theta (L_\theta^S)^2.
\end{align*}
When the range of $\rho_\theta$ contains the range of 
$(\frac{d \rho_\theta}{d \theta})^2$,
the RLD Fisher information has another expression:
\begin{align}
J^R_\theta= \Tr (\frac{d \rho_\theta}{d \theta})^2 \rho_\theta^{-1}.
\Label{20-11-a}
\end{align}
When the state family $\{\rho_\theta\}$
is given by $\rho_\theta:=e^{i \theta H}|u\rangle \langle u|
e^{-i \theta H}$, 
the condition (\ref{20-11-a}) does not hold,
where $H$ is an Hermitian matrix.
In this case, the SLD Fisher information is calculated as follows\cite{FN}.
\begin{align}
4(\langle u|H^2|u\rangle -\langle u|H|u\rangle^2).\Label{21-1}
\end{align}

Now, we introduce the unbiased condition by
\begin{align*}
\int \hat{\theta} \Tr \rho_\theta^{\otimes n} M^n(d \hat{\theta})=\theta, \quad
\forall \theta\in \Theta.
\end{align*}
However, this condition is sometimes too restrictive in the asymptotic setting.
So, we consider the Taylor expansion at a point $\theta_0$ and focus on 
the first order.
Then, we obtain the locally unbiased condition at $\theta_0$:
\begin{align*}
\int \hat{\theta} \Tr \rho_{\theta_0}^{\otimes n} M^n(d \hat{\theta})=\theta_0, \quad
\frac{d}{d \theta}\int \hat{\theta} \Tr \rho_\theta^{\otimes n}
M^n(d \hat{\theta})|_{\theta=\theta_0}=1.
\end{align*}
Under the locally unbiased condition at $\theta_0$,
an application of Schwarz inequality similar to
the classical case 
yields the quantum Cram\'{e}r-Rao inequalities for both quantum Fisher information.
\begin{align}
\MSE_{\theta_0} (M^n) & \ge \frac{1}{n} (J^R_{\theta_0})^{-1}\Label{12-1-2} \\
\MSE_{\theta_0} (M^n) & \ge \frac{1}{n} (J^S_{\theta_0})^{-1}\Label{12-1-1}.
\end{align}
Since $J^R_\theta$ is greater than $J^S_\theta$,
the inequality (\ref{12-1-1})
is more informative than the inequality (\ref{12-1-2}).
When the estimator $M^n$ is the spectral decomposition of the operator
$ \theta_0 I+ \frac{1}{n J^S_{\theta_0}}(L_{\theta_0}^{S,(1)}+\cdots+ L_{\theta_0}^{S,(n)})$,
the equality in (\ref{12-1-1}) holds,
where $X^{(j)}$ is given as $I^{\otimes j-1} \otimes X\otimes I^{\otimes n-j}$.
Then, we obtain the following inequality
\begin{align*}
J^R_\theta \ge J^S_\theta.
\end{align*}
This inequality seems to imply that $J^R_\theta$ is not as meaningful as $J^S_\theta$
in the one-parametric case.
However, as will be explained latter, 
$J^R_\theta$ provides a meaningful bound for MSE in the case of channel estimation.

In fact, in the asymptotic setting,
a suitable estimator usually satisfies the asymptotic locally unbiased condition:
\begin{align*}
\lim_{n \to \infty}
\int \hat{\theta} \Tr \rho_{\theta_0}^{\otimes n} M^n(d \hat{\theta})=\theta_0, \quad
\lim_{n \to \infty}
\left. \frac{d}{d \theta}
\int \hat{\theta} \Tr \rho_\theta^{\otimes n} M^n(d \hat{\theta})
\right|_{\theta=\theta_0}=1
\end{align*}
for all points $\theta_0$.
Under the above condition,
using (\ref{12-1-1}), 
we obtain the inequality
\begin{align}
\limsup_{n \to\infty} n \MSE_\theta(M^n)
\ge
(J^S_\theta)^{-1}.\Label{22-1}
\end{align}
Further, by using the two-step method,
the bound $(J^S_\theta)^{-1}$
can be universally attained for any true parameter $\theta$ \cite{4,5}.
So, defining the Cram\'{e}r-Rao bound:
\begin{align*}
C_\theta:=
\inf_{\{M^n\}} \left\{\limsup_{n \to\infty} n \MSE_\theta(M^n)
\left|
\begin{array}{l}
\{M^n\}\hbox{ satisfies the asymptotic}\\
\hbox{ locally unbiased condition.}
\end{array}
\right.
\right\},
\end{align*}
we obtain
\begin{align*}
C_\theta=(J^S_\theta)^{-1}.
\end{align*}

On the other hand, its multi-parameter case is much complicated even in the asymptotic setting \cite{q-cent,GK,GJ,GJK,spin}.
So, this paper does not treat the multi-parameter case.

\section{Maximum SLD and RLD Fisher informations in quantum channel estimation}\Label{s3}
In this section, we apply the Cram\'{e}r-Rao approach to estimation of channel.
In the quantum system,  
the channel is given by a trace preserving completely positive (TP-CP) map $\Lambda$
from the set of densities on the input system ${\cal H}:=\C^d$
to the set of densities on the output system ${\cal K}:=\C^{d'}$.
By using $dd'$ linear maps $F_i$ from ${\cal S}({\cal H})$ to 
${\cal S}({\cal K})$,
any TP-CP map $\Lambda$ can be described by 
$\Lambda(\rho)=\sum_{i=1}^{dd'}F_i \rho F_i^\dagger$.
Hence, our task is to estimate the true TP-CP map under the assumption that 
the true TP-CP map belongs to a certain family of TP-CP maps
$\{\Lambda_\theta\}$.

In order to characterize a TP-CP map $\Lambda_\theta$,
we formulate the notation concerning states on the tensor product system 
${\cal H}\otimes {\cal R}$,
where ${\cal R}$ is a system of the same dimensionality as
${\cal H}$ and 
is called the reference system.
Using a linear map $A$ from ${\cal R}$ to ${\cal H}$,
we define an element $|A\rangle\rangle$ of ${\cal H}\otimes {\cal R}$
as follows.
\begin{align*}
|A \rangle \rangle := \sum_{j,k}
A_{j,k} |j\rangle_{H}\otimes|k\rangle_{R},
\end{align*}
where 
$\{|j\rangle_{H}\}_{j=1, \ldots, d}$ 
and
$\{|k\rangle_{R}\}_{k=1, \ldots, d}$ 
are complete orthonormal systems (CONSs) of ${\cal H}$ and ${\cal R}$.
Hence, the relation 
\begin{align*}
B\otimes C |A \rangle \rangle 
=|B A C^T \rangle \rangle
\end{align*}
holds.
This notation is applied to the cases of 
${\cal K}\otimes {\cal H}$ and ${\cal K}\otimes {\cal R}$.

Now, we focus on the matrix
$\rho[\Lambda_\theta]:=(\Lambda_\theta\otimes \id) (|I \rangle \rangle \langle \langle I|)$,
which is called the Choi-Jamiolkowski representation matrix \cite{Choi,Jami}.
Then, 
when the input state is the maximally entangled state 
$
|\frac{1}{\sqrt{d}}I \rangle\rangle=
\sum_{j=1}^d\frac{1}{\sqrt{d}} |j\rangle\otimes |j\rangle $,
the output state 
is $\frac{1}{d}\rho[\Lambda_\theta]
=(\Lambda_\theta\otimes \id) (| \frac{1}{\sqrt{d}}I \rangle \rangle 
\langle\langle \frac{1}{\sqrt{d}}I|)$.
When the matrix $\overline{A} A^T$ is a density matrix on 
${\cal H}$,
$|A \rangle \rangle \langle \langle A |$
is a pure state on the product system ${\cal H}\otimes {\cal R}$.
Thus, the output state is given as
\begin{align*}
(\Lambda_\theta\otimes \id) (| A \rangle \rangle \langle \langle A |)
=
(I\otimes A^T)(\Lambda_\theta\otimes \id) (| I \rangle \rangle 
\langle\langle I|)(I\otimes \overline{A})
=
(I\otimes A^T)\rho[\Lambda_\theta](I\otimes \overline{A}).
\end{align*}
In the one-parameter case, 
we express the derivative $\frac{d \rho(\Lambda_\theta)}{d \theta}$ by 
$D[\Lambda_\theta]$.

When the input state is the product state 
$|v\rangle\langle v| \otimes |u\rangle \langle u|$,
the output state is
$\Lambda_\theta(|v\rangle\langle v|) \otimes |u\rangle \langle u|
=
( I  \otimes u \cdot v^{T} )\rho[\Lambda_\theta] 
( I  \otimes \overline{v} \cdot u^{\dagger} )$.
Since
\begin{align*}
& \langle \overline{v} |
\Tr_{{\cal K}}
\rho[\Lambda_\theta] 
|\overline{v} \rangle 
|u\rangle \langle u|
=
u \cdot v^{T} 
\Tr_{{\cal K}}
\rho[\Lambda_\theta] 
\overline{v} \cdot u^{\dagger} \\
=&
\Tr_{{\cal K}}( I  \otimes u \cdot v^{T} )\rho[\Lambda_\theta] 
( I  \otimes \overline{v} \cdot u^{\dagger} )
=
|u\rangle \langle u|,
\end{align*}
we have
\begin{align*}
\langle \overline{v}| 
(\Tr_{{\cal K}}\rho[\Lambda_\theta])  | \overline{v}\rangle
=1.
\end{align*}
Thus, we obtain
\begin{align}
\Tr_{{\cal K}}\rho[\Lambda_\theta]
=I.\Label{20-4}
\end{align}
Taking the derivative in (\ref{20-4}),
we obtain
\begin{align}
\Tr_{{\cal K}}
D[\Lambda_\theta]
=0.\Label{20-5}
\end{align}

Now, we back to our estimation problem.
In this problem, our choice is given by a pair of the input state 
$\rho$ and the quantum measurement $M$.
When we fix the input state, our estimation problem can be reduced to
the state estimation with the state family $\{\Lambda_\theta(\rho)|\theta
\in \Theta\}$.
In the one-parameter case,
we focus on the suprema
\begin{align*}
J^R[\Lambda_\theta]:=
\sup_{\rho} 
J^R[\Lambda_\theta,\rho], \quad
J^S[\Lambda_\theta]:=
\sup_{\rho} 
J^S[\Lambda_\theta,\rho],
\end{align*}
where
$J^S[\Lambda_\theta,\rho]$ and $J^R[\Lambda_\theta,\rho]$
are the SLD and RLD Fisher informations when 
the input state is $\rho$.
In particular,
it is important to calculate the supremum $J^S[\Lambda_\theta]$
which is smaller than $J^R[\Lambda_\theta]$.

When $n$ applications of the unknown channel $\Lambda_\theta$
are available, the input state $\rho_n$ 
and the measurement $M^n$ are given as a state on 
$({\cal H}\otimes {\cal R})^{\otimes n}$ 
and a POVM on 
$({\cal K}\otimes {\cal R})^{\otimes n}$.
For a sequence of estimators $\{ (\rho_n,M^n)\}$,
we consider the asymptotic locally unbiased condition:
\begin{align*}
\lim_{n \to \infty}
\int \hat{\theta} \Tr \Lambda_{\theta_0}(\rho_n) M^n(d \hat{\theta})=\theta_0, \quad
\lim_{n \to \infty}
\left. \frac{d}{d \theta}
\int \hat{\theta} \Tr \Lambda_{\theta}(\rho_n)  M^n(d \hat{\theta})
\right|_{\theta=\theta_0}=1
\end{align*}
for all points $\theta_0$,
and denotes the MSE of $(\rho_n,M^n)$
by $\MSE_\theta(\rho_n,M^n)$.
Assume that
$J^S[\Lambda_\theta^{\otimes n}]$
behaves as $O(n^{\alpha})$ when $n$ goes to infinity.
When 
$\{(\rho_n,M^n)\}$ satisfies the asymptotic
locally unbiased condition,
the inequality (\ref{22-1}) yields that
\begin{align*}
\limsup_{n \to \infty}
n^\alpha \MSE_\theta(\rho_n,M^n)
\ge 
\limsup_{n \to \infty} 
\frac{n^\alpha}{J^S[\Lambda_\theta^{\otimes n}]}.
\end{align*}
We define the Cram\'{e}r-Rao bound:
\begin{align*}
&\tilde{C}_\alpha [\Lambda_\theta]\\
:=&
\inf_{\{\rho_n, M^n\}} \left\{\limsup_{n \to\infty} n^\alpha \MSE_\theta(\rho_n,M^n)
\left|
\begin{array}{l}
\{(\rho_n,M^n)\}\hbox{ satisfies the asymptotic}\\
\hbox{ locally unbiased condition.}
\end{array}
\right.
\right\}.
\end{align*}
Thus, 
we obtain
\begin{align}
\tilde{C}_\alpha [\Lambda_\theta]=
\limsup_{n \to \infty} 
\frac{n^\alpha}{J^S[\Lambda_\theta^{\otimes n}]}.
\label{3-20-1}
\end{align}
Since 
\begin{align}
J^S[\Lambda_\theta^{\otimes n+m}]
\ge J^S[\Lambda_\theta^{\otimes n}]+ J^S[\Lambda_\theta^{\otimes m}], 
\Label{4-23-2}
\end{align}
the limit $\lim_{n \to \infty} \frac{J^S[\Lambda_\theta^{\otimes n}]}{n}$ exists.
 (For example, see Lemma A.1 of \cite{3}.)
Thus, $\tilde{C}_1 [\Lambda_\theta]$ can be defined by 
$\lim_{n \to \infty} \frac{n^\alpha}{J^S[\Lambda_\theta^{\otimes n}]}$.

In order to treat the above values,
we consider the following condition:
\begin{description}
\item[(C)]
The range of $\rho[\Lambda_\theta]$
contains the range of $D[\Lambda_\theta]^2$.
\end{description}
Assume that the condition (C) does not hold.
When the input state is the maximally entangled state 
$|\frac{1}{\sqrt{d}}I \rangle\rangle$,
the RLD Fisher information diverges.
So, $J^R[\Lambda_\theta]$ is infinity.

\begin{thm}\Label{th1}
When the condition (C) holds,
\begin{align*}
J^R[\Lambda_\theta]
=
\|\Tr_{{\cal K}} 
D[\Lambda_\theta]
\rho[\Lambda_\theta]^{-1}
D[\Lambda_\theta] \|.
\end{align*}
\end{thm}

\begin{proof}
Assume that the input state is given by
$| A \rangle \rangle \langle \langle A |$
and $A$ is an invertible matrix.
Then, the range of $(I\otimes A^T)\rho[\Lambda_\theta](I\otimes \overline{A})$
contains the range of $((I\otimes A^T)D[\Lambda_\theta](I\otimes \overline{A}))^2$.

Assume that $\rho[\Lambda_\theta] $ is invertible.
Using the formula (\ref{20-11-a}), we obtain
\begin{align*}
& J^R[\Lambda_\theta,| A \rangle \rangle \langle \langle A |]=
\Tr((I\otimes A^T)D[\Lambda_\theta](I\otimes \overline{A}))^2
((I\otimes A^T)\rho[\Lambda_\theta](I\otimes \overline{A}))^{-1} \\
= &
\Tr
(I\otimes A^T)D[\Lambda_\theta](I\otimes \overline{A})
(I\otimes A^T)D[\Lambda_\theta](I\otimes \overline{A})
(I\otimes \overline{A})^{-1}
\rho[\Lambda_\theta]^{-1}
(I\otimes A^T)^{-1}\\
= &
\Tr
(I\otimes \overline{A} A^T)D[\Lambda_\theta] \rho[\Lambda_\theta]^{-1}
D[\Lambda_\theta]\\
= &
\Tr \overline{A} A^T
(\Tr_{{\cal K}}
D[\Lambda_\theta] \rho[\Lambda_\theta]^{-1}
D[\Lambda_\theta]),
\end{align*}
where
$((I\otimes A^T)\rho[\Lambda_\theta](I\otimes \overline{A}))^{-1}$
is the inverse of
$(I\otimes A^T)\rho[\Lambda_\theta](I\otimes \overline{A})$ on its range.
So, the supremum of $\Tr \overline{A} A^T
(\Tr_{{\cal K}} D[\Lambda_\theta] \rho[\Lambda_\theta]^{-1}
D[\Lambda_\theta])$ 
with the condition $\rank \overline{A} A^T= \dim {\cal H}$
equals $\|\Tr_{{\cal K}}
D[\Lambda_\theta] \rho[\Lambda_\theta]^{-1}
D[\Lambda_\theta]\|$.
Assume that $\rho[\Lambda_\theta] $ is non-invertible.
We choose an arbitrary small real number $\epsilon>0$.
Similar calculations 
and the operator monotonicity of $x \mapsto -x^{-1}$
yield that
\begin{align}
&
\Tr((I\otimes A^T)D[\Lambda_\theta](I\otimes \overline{A}))^2
((I\otimes A^T)(\rho[\Lambda_\theta]+\epsilon I)(I\otimes \overline{A}))^{-1} \nonumber \\
\le &
\Tr \overline{A} A^T
(\Tr_{{\cal K}}
D[\Lambda_\theta] (\rho[\Lambda_\theta]+\epsilon I)^{-1}
D[\Lambda_\theta]) \nonumber \\
\le &
\Tr \overline{A} A^T
(\Tr_{{\cal K}}
D[\Lambda_\theta] \rho[\Lambda_\theta]^{-1}
D[\Lambda_\theta]),\label{8-2-1}
\end{align}
which is a bounded value due to Condition (C).
Now, we focus on two matrixes on the range of 
$(I\otimes A^T)\rho[\Lambda_\theta](I\otimes \overline{A})$, 
$
(I\otimes A^T)\rho[\Lambda_\theta](I\otimes \overline{A})$
and 
$((I\otimes A^T)D[\Lambda_\theta](I\otimes \overline{A}))^2$.
Since the right hand side of (\ref{8-2-1}) is independent of $\epsilon$,
the range of 
$(I\otimes A^T)\rho[\Lambda_\theta](I\otimes \overline{A})$
contains 
that of 
$((I\otimes A^T)D[\Lambda_\theta](I\otimes \overline{A}))^2$.
So, taking the limit $\epsilon\to 0$, we obtain
\begin{align*}
& J^R[\Lambda_\theta,| A \rangle \rangle \langle \langle A |]
=
\Tr
((I\otimes A^T)D[\Lambda_\theta](I\otimes \overline{A}))^2
((I\otimes A^T)\rho[\Lambda_\theta](I\otimes \overline{A}))^{-1} \\
= &
\lim_{\epsilon \to +0} 
\Tr((I\otimes A^T)D[\Lambda_\theta](I\otimes \overline{A}))^2
((I\otimes A^T)(\rho[\Lambda_\theta]+\epsilon I)(I\otimes \overline{A}))^{-1} \\
\le &
\Tr \overline{A} A^T
(\Tr_{{\cal K}}
D[\Lambda_\theta] \rho[\Lambda_\theta]^{-1}
D[\Lambda_\theta]).
\end{align*}

The remaining task is to show the inequality
\begin{align}
J^R[\Lambda_\theta,| A \rangle \rangle \langle \langle A |]
\le
\|\Tr_{{\cal K}} D[\Lambda_\theta] \rho[\Lambda_\theta]^{-1} D[\Lambda_\theta]\|
\label{3-8}
\end{align}
for a non-invertible matrix $A$. 
Define $(I \otimes A^T)^{-1}$ and $(I\otimes \overline{A})^{-1}$
as the inverses of $I \otimes A^T$ and $I\otimes \overline{A}$
whose domains are the ranges of the matrixes $I \otimes A^T$ and $I \otimes \overline{A}$.
Letting $P$ and $P'$ be
the projections to the ranges of the matrixes $I \otimes A^T$ and $I \otimes \overline{A}$,
we have
$(I \otimes A^T)=P (I \otimes A^T) = (I \otimes A^T)P'=P(I \otimes A^T)P'$, and
$(I \otimes \overline{A})=P' (I \otimes \overline{A}) = (I \otimes \overline{A})P=P'(I \otimes \overline{A})P$.
Then,
$(I\otimes A^T)\rho[\Lambda_\theta](I\otimes \overline{A})
=P(I\otimes A^T)\rho[\Lambda_\theta](I\otimes \overline{A})P
=P(I\otimes A^T)P'\rho[\Lambda_\theta]P'(I\otimes \overline{A})P$.

In the following, we consider the case when 
$\rho[\Lambda_\theta]$ is invertible.
The matrix $(P' \rho[\Lambda_\theta] P')^{-1}$ is defined as
the inverse of $P' \rho[\Lambda_\theta] P'$ whose domain and range are the range of $P'$.
Thus,
\begin{align}
& J^R[\Lambda_\theta,| A \rangle \rangle \langle \langle A |]=
\Tr P((I\otimes A^T)D[\Lambda_\theta](I\otimes \overline{A}))^2 P
(P(I\otimes A^T)\rho[\Lambda_\theta](I\otimes \overline{A})P)^{-1} \nonumber\\
= &
\Tr
 (I\otimes A^T) D[\Lambda_\theta] P'(I\otimes \overline{A})
P
(I\otimes \overline{A})^{-1} 
(P' \rho[\Lambda_\theta] P')^{-1}
(I\otimes A^T)^{-1}P
(I\otimes A^T) P' D[\Lambda_\theta] (I\otimes \overline{A})\nonumber
\\
= &
\Tr
 (I\otimes A^T)  D[\Lambda_\theta] 
P'(P' \rho[\Lambda_\theta] P')^{-1}P'
 D[\Lambda_\theta]  (I\otimes \overline{A})\nonumber \\
= &
\Tr
(P' \rho[\Lambda_\theta] P')^{-1}
P' D[\Lambda_\theta]  (I\otimes \overline{A} A^T)  D[\Lambda_\theta] P'
\nonumber \\
= &
 \Tr
(P' \rho[\Lambda_\theta] P')^{-1}
 D[\Lambda_\theta]  (I\otimes \overline{A} A^T)  D[\Lambda_\theta] 
\nonumber \\
\le & 
\Tr
\rho[\Lambda_\theta]^{-1} 
 D[\Lambda_\theta]  (I\otimes \overline{A} A^T)  D[\Lambda_\theta] 
\label{3-12}\\
= &
\Tr \overline{A} A^T
(\Tr_{{\cal K}}
D[\Lambda_\theta] \rho[\Lambda_\theta]^{-1}
D[\Lambda_\theta])\nonumber .
\end{align}
Therefore, 
the inequality (\ref{3-8}) holds.

Next, we consider the case when 
$\rho[\Lambda_\theta]$ is non-invertible.
We choose an arbitrary small real number $\epsilon >0$.
Similar calculations 
and the operator monotonicity of $x \mapsto -x^{-1}$
yield that
\begin{align}
& 
\Tr P((I\otimes A^T)D[\Lambda_\theta](I\otimes \overline{A}))^2 P
(P(I\otimes A^T)
(\rho[\Lambda_\theta] +\epsilon I)
(I\otimes \overline{A})P)^{-1} \nonumber\\
\le &
\Tr 
\overline{A} A^T
(\Tr_{{\cal K}}
D[\Lambda_\theta] (\rho[\Lambda_\theta]+\epsilon I)^{-1}
D[\Lambda_\theta])
\nonumber\\
\le & 
\Tr 
\overline{A} A^T
(\Tr_{{\cal K}}
D[\Lambda_\theta] \rho[\Lambda_\theta]^{-1}
D[\Lambda_\theta]), 
\label{8-1-1}
\end{align}
which is a bounded value due to Condition (C).
Now, we focus on two matrixes on the range of $P$, 
$P(I\otimes A^T)\rho[\Lambda_\theta]
(I\otimes \overline{A})P$
and
$P((I\otimes A^T)D[\Lambda_\theta](I\otimes \overline{A}))^2 P$.
Since the right hand side of (\ref{8-1-1}) is independent of $\epsilon$,
the range of 
$P(I\otimes A^T)\rho[\Lambda_\theta]
(I\otimes \overline{A})P$
contains 
that of 
$P((I\otimes A^T)D[\Lambda_\theta](I\otimes \overline{A}))^2 P$.
So, taking the limit $\epsilon\to 0$, we obtain
\begin{align}
& J^R[\Lambda_\theta,| A \rangle \rangle \langle \langle A |]=
\Tr P((I\otimes A^T)D[\Lambda_\theta](I\otimes \overline{A}))^2 P
(P(I\otimes A^T)\rho[\Lambda_\theta](I\otimes \overline{A})P)^{-1} \nonumber\\
= &
\lim_{\epsilon \to +0} 
\Tr P((I\otimes A^T)D[\Lambda_\theta](I\otimes \overline{A}))^2 P
(P(I\otimes A^T)
(\rho[\Lambda_\theta] +\epsilon I)
(I\otimes \overline{A})P)^{-1} 
\nonumber  \\
\le & 
\Tr \overline{A} A^T
(\Tr_{{\cal K}}
D[\Lambda_\theta] \rho[\Lambda_\theta]^{-1}
D[\Lambda_\theta])\nonumber .
\end{align}
Therefore, 
the inequality (\ref{3-8}) holds.
\end{proof}

\begin{thm}\Label{th2}
When 
two channel families 
$\{\Lambda_{\theta}\}$
and $\{\tilde{\Lambda}_{\theta}\}$
satisfy the condition (C),
then the additivity
\begin{align*}
J^R[\Lambda_\theta\otimes \tilde{\Lambda}_{\theta}]
=
J^R[\Lambda_\theta]
+J^R[\tilde{\Lambda}_\theta]
\end{align*}
holds.
\end{thm}

\begin{proof}
Let ${\cal K}$ and $\tilde{{\cal K}}$ be output systems of 
the channels $\Lambda_{\theta}$ and $\tilde{\Lambda}_{\theta}$.
Then, the relation (\ref{20-5}) guarantees 
\begin{align}
\Tr_{{\cal K}\otimes \tilde{{\cal K}}}
D[\Lambda_\theta] \otimes
D[\tilde{\Lambda}_\theta]
=
\Tr_{{\cal K}}
D[\Lambda_\theta] 
\Tr_{\tilde{{\cal K}}}
D[\tilde{\Lambda}_\theta]
=0.\Label{4-7-2}
\end{align}
Since $D[\Lambda_\theta\otimes \tilde{\Lambda}_{\theta}]
=
D[\Lambda_\theta]\otimes \rho[\tilde{\Lambda}_{\theta}]
+\rho[\Lambda_\theta]\otimes D[\tilde{\Lambda}_{\theta}]$,
Theorem \ref{th1} and (\ref{4-7-2}) yield 
\begin{align}
& J^R[\Lambda_\theta\otimes \tilde{\Lambda}_{\theta}]\nonumber \\
= &
\|
\Tr_{{\cal K}\otimes \tilde{{\cal K}}}
(D[\Lambda_\theta]\otimes \rho[\tilde{\Lambda}_{\theta}]
+\rho[\Lambda_\theta ]\otimes D[\tilde{\Lambda}_{\theta}])
(\rho[\Lambda_\theta ]^{-1}\otimes \rho[\tilde{\Lambda}_{\theta}]^{-1})
(D[\Lambda_\theta]\otimes \rho[\tilde{\Lambda}_{\theta}]
+\rho[\Lambda_\theta ] \otimes D[\tilde{\Lambda}_{\theta}])
\| \nonumber \\
= &
\|
\Tr_{{\cal K}\otimes \tilde{{\cal K}}}
(D[\Lambda_\theta] \rho[\Lambda_\theta ]^{-1} D[\Lambda_\theta]
\otimes \rho[\tilde{\Lambda}_{\theta}]
+
\rho[\Lambda_\theta ]\otimes
D[\tilde{\Lambda}_\theta] \rho[\tilde{\Lambda}_{\theta}]^{-1} D[\tilde{\Lambda}_\theta]
+
2D[\Lambda_\theta]
\otimes
D[\tilde{\Lambda}_\theta] )\| \nonumber \\
=& 
\|
(\Tr_{{\cal K}}
D[\Lambda_\theta]
\rho[\Lambda_\theta ]^{-1}
D[\Lambda_\theta])\otimes I
+
I \otimes
(\Tr_{\tilde{{\cal K}}} D[\tilde{\Lambda}_\theta] \rho[\tilde{\Lambda}_{\theta}]^{-1}
D[\tilde{\Lambda}_\theta]) \| \nonumber  \\
=&
\|
\Tr_{{\cal K}}
D[\Lambda_\theta]
\rho[\Lambda_\theta ]^{-1}
D[\Lambda_\theta]\|
+
\|
\Tr_{\tilde{{\cal K}}}
D[\tilde{\Lambda}_\theta]\rho[\tilde{\Lambda}_{\theta}]^{-1}D[\tilde{\Lambda}_\theta] \| \Label{4-7-1} \\
=&
J^R[\Lambda_\theta]
+J^R[\tilde{\Lambda}_\theta],\nonumber 
\end{align}
where the equation (\ref{4-7-1}) follows from the 
additivity property concerning matrix norm:
\begin{align*}
\|X\otimes I + I \otimes Y\|
=\|X\|+\|Y\|
\end{align*}
for any two Hermitian matrixes $X$ and $Y$.
\end{proof}

\begin{cor}
When 
a channel family
$\{\Lambda_{\theta}\}$
satisfies the condition (C),
then 
\begin{align*}
J^R[\Lambda_\theta^{\otimes n}]=
n
J^R[\Lambda_\theta].
\end{align*}
\end{cor}

Since $n J^S[\Lambda_\theta] \le J^S[\Lambda_\theta^{\otimes n}]
\le J^R[\Lambda_\theta^{\otimes n}]$,
$J^S[\Lambda_\theta^{\otimes n}]$ increases in order $n$ under the assumption of Theorem \ref{th1}, i.e., 
$J^S[\Lambda_\theta^{\otimes n}]=O(n)$.
When the rank of $\rho[\Lambda_\theta]$ is the maximum, i.e., $dd'$,
this condition holds and 
$J^S[\Lambda_\theta^{\otimes n}]=O(n)$.
However, there is an example that 
does not satisfy the above condition but satisfies
the condition (C) as follows.
So, the condition (C) is weaker than the condition that
$\rho[\Lambda_\theta]$ has the maximum rank.

A channel $\Lambda$ is called a phase damping channel
when 
the output system ${\cal K}$ equals the input system ${\cal H}$ and
there exist complex numbers $d_{k,l}$ such that
\begin{align*}
\Lambda_{\{d_{k,l} \}}(\rho)=\sum_{k,l} d_{k,l}\rho_{k,l}|k\rangle \langle l| ,
\end{align*}
where $\rho=\sum_{k,l} \rho_{k,l}|k\rangle \langle l|$.
In this case, the state $\rho[\Lambda]$ is written as the following form
\begin{align*}
\rho[\Lambda_{\{d_{k,l} \}}]=
\sum_{k,l} d_{k,l}|k\rangle|k\rangle \langle l|\langle l| .
\end{align*}
That is, the range of $\rho[\Lambda]$ is included by 
the space spanned by
$\{|k\rangle_{K}|k\rangle_{R}\}$.
When a channel family $\{\Lambda_\theta\}$ 
is given as a one-parameter subfamily of
$\{ \Lambda_{\{d_{k,l} \}} | d_{k,l}\hbox{ is strictly positive.}  \}$,
the condition (C) holds.
Therefore, there exists a channel family that satisfies the condition (C) but
consists of non-full-rank channels.

Further, we have the following observation.
\begin{cor}
Assume that the condition (C) holds
and there exists a normalized vector 
$u$ in the input system ${\cal H}$ such that
$\|\Tr_{{\cal K}}
D[\Lambda_\theta] \rho[\Lambda_\theta]^{-1}
D[\Lambda_\theta]\|
=\langle u| \Tr_{{\cal K}}
D[\Lambda_\theta] \rho[\Lambda_\theta]^{-1}
D[\Lambda_\theta]|u \rangle$,
$[\langle u| \rho[\Lambda_\theta]|u \rangle ,
\langle u| D[\Lambda_\theta]|u \rangle]=0$, and
$[I_{{\cal K}}\otimes |u \rangle\langle u|,
\rho[\Lambda_\theta]]=0$.
Then,
$J^S[\Lambda_\theta] = J^R[\Lambda_\theta]$,
and this bound can be attained 
by the input pure state $|u \rangle\langle u|$ on ${\cal H}$.
That is, it can be attained without use of the reference system.
\end{cor}

\begin{proof}
Since
$[\langle u| \rho[\Lambda_\theta]|u \rangle ,
\langle u| D[\Lambda_\theta]|u \rangle]=0$, 
$J^S[\Lambda_\theta,|u \rangle\langle u|] = J^R[\Lambda_\theta,|u \rangle\langle u|]$.
So, it is sufficient to show that $
J^R[\Lambda_\theta,|u \rangle\langle u|]
=\langle u| \Tr_{{\cal K}}
D[\Lambda_\theta] \rho[\Lambda_\theta]^{-1}
D[\Lambda_\theta]|u \rangle$.

We consider the case when the input state is 
$|u \rangle\langle u|\otimes |v \rangle\langle v|$.
Remember that 
$J^R[\Lambda_\theta,|u \rangle\langle u|\otimes |v \rangle\langle v|]
=
J^R[\Lambda_\theta,|u \rangle\langle u|]$.
Since 
$[I_{{\cal K}}\otimes |u \rangle\langle u|,
\rho[\Lambda_\theta]]=0$,
the equality in (\ref{3-12}) holds for $A=u \cdot v^T$.
Then,
\begin{align*}
& J^R[\Lambda_\theta,|u \rangle\langle u|\otimes |v \rangle\langle v|]
=
\Tr \overline{(u \cdot v^T)} (u \cdot v^T)^{T} 
\Tr_{{\cal K}}
D[\Lambda_\theta] \rho[\Lambda_\theta]^{-1}D[\Lambda_\theta] \\
=& \langle u| \Tr_{{\cal K}}
D[\Lambda_\theta] \rho[\Lambda_\theta]^{-1}D[\Lambda_\theta]|u \rangle.
\end{align*}

\end{proof}

\section{Examples}\Label{s31}
We consider the case where the condition (C) does not hold.
As the simplest example, we consider the one-parameter unitary case, i.e., 
the case when 
$\Lambda_\theta(\rho)=e^{i \theta H} \rho e^{-i\theta H}$ with 
an Hermitian matrix $H$.
Using (\ref{21-1}),
we obtain
\begin{align*}
J^S[\Lambda_\theta]
=(\lambda_{\max}(H)-\lambda_{\min}(H))^2,
\end{align*}
where 
$\lambda_{\max}(H)$ and $\lambda_{\min}(H)$
are the maximum and minimum of eigenvalues of $H$.
So, we obtain 
\begin{align*}
J^S[\Lambda_\theta^{\otimes n}]
=n^2 (\lambda_{\max}(H)-\lambda_{\min}(H))^2.
\end{align*}
In particular,
in the two-dimensional case,
when
$H=\left(
\begin{array}{cc}
\frac{1}{2} & 0 \\
0 & -\frac{1}{2} 
\end{array}
\right)$,
the optimal input is 
$\frac{1}{\sqrt{2}}(|0\rangle ^{\otimes n}+ |1\rangle^{\otimes n})$,
which is called the noon state.
This unitary estimation is called phase estimation
and this estimation with the noon state
was experimentally realized with $n=4$\cite{take1,take2} and $n=10$\cite{Mag}.

Next, we consider the $d$-dimensional system $\complex^d$ spanned by $\{|j\rangle\}_{j=0}^{d-1}$
and the unitary matrix $X$ defined as $X|j\rangle:=|j+1\rangle$ mod $d$.
Using a distribution $\{p_{j}\}_{j=0}^{d-1}$
and a real diagonal matrix $H$ with diagonal elements $\{h_{j}\}_{j=0}^{d-1}$,
we define the TP-CP map $\Lambda_\theta$ by
\begin{align*}
\Lambda_\theta(\rho):=   
\sum_{j=0}^{d-1} p_j X^j e^{i \theta H}\rho e^{- i \theta H} X^{-j}.
\end{align*}
This TP-CP map can be regarded as 
the stochasitc application of the unitary $X^j$ after the application of the unitary $e^{i \theta H}$.

Let $h_a$ and $h_b$ be the maximum and the minimum eigenvalues of $H$.
Using two-dimensional reference system spanned by $|0\rangle_R$ and $|1\rangle_R$,
we choose the following input state:
\begin{align*}
|\Phi_n\rangle:=
\frac{1}{\sqrt{2}}
(
|a\rangle^{\otimes n}|0\rangle_R 
+  |b\rangle^{\otimes n}|1\rangle_R).
\end{align*}
In this case, 
as the first step, we apply the following measurement 
$\{M_{\vec{j}}\}$:
\begin{align*}
M_{\vec{j}}:=& 
|\vec{j}+(a,\ldots, a)\rangle \langle  \vec{j}+(a,\ldots, a)|
\otimes |0\rangle_R ~_R \langle 0| \\
&+
|\vec{j}+(b,\ldots, b)\rangle \langle  \vec{j}+(b,\ldots, b)|
\otimes |1\rangle_R ~_R \langle 1|.
\end{align*}
When the outcome of this measurement is $\vec{j}$,
the resulting state is the pure state
\begin{align}
\frac{1}{\sqrt{2}}
(
e^{i n h_a }|\vec{j}+(a,\ldots, a)\rangle |0\rangle_R 
+ e^{i n h_b } |\vec{j}+(b,\ldots, b)\rangle |1\rangle_R).
\end{align}
The SLD Fisher information of the above family is
$n^2(h_a-h_b)^2$.
Therefore, since
the maximum SLD Fisher information behaves as $O(n^2)$ at most,
$J^S[\Lambda_0^{\otimes n}]$ behaves as $O(n^2)$.

\section{Local asymptotic mini-max bound}\Label{s4}
In this section, we consider the relation between the discussion in the previous section and estimating protocols in a different viewpoint.
Consider the phase estimation with inputing 
the noon state 
$\frac{1}{\sqrt{2}}(|0\rangle ^{\otimes n}+ |1\rangle^{\otimes n})$.
Then, the output state is 
$\frac{1}{\sqrt{2}}(e^{in \theta/2}|0\rangle ^{\otimes n}
+ e^{-in \theta/2}|1\rangle^{\otimes n})$.
In this case, we cannot distinguish the parameters $\theta$ and $\theta + \frac{2\pi}{n}$.
For example, when we apply measurement 
$\{\frac{1}{\sqrt{2}}(|0\rangle ^{\otimes n}\pm |1\rangle^{\otimes n})\}$,
the probability with the outcome $\frac{1}{\sqrt{2}}(|0\rangle ^{\otimes n}+ |1\rangle^{\otimes n})$
equals $\cos^2 n\theta/2 $, as is shown in Fig \ref{f3},
and the Fisher information equals $n^2$.
Even if the parameter $\theta$ is assumed to belong to $(0,\pi/n]$,
we cannot distinguish the two parameters $\theta=\pi/3n$ and $\theta= 2\pi/3n$ 
with so high probability
because we have only two outcomes.
In order to distinguish two parameters $\theta=\pi/3n$ and $\theta= 2\pi/3n$ in this measurement,
we need to repeat this measurement with several times, e.g., $k$.
Since the number of application of the unknown unitary is $N:=kn$,
the error behaves as $\frac{1}{k}\frac{1}{N}$,
which is different from $O(\frac{1}{N^2})$.
So, we cannot conclude that the above method
attains the order $O(\frac{1}{N^2})$ concerning MSE.
Therefore, we need to discuss what a bound can be attained globally,
more carefully.

\begin{figure}[htbp]
\begin{center}
\scalebox{1.0}{\includegraphics[scale=1.4]{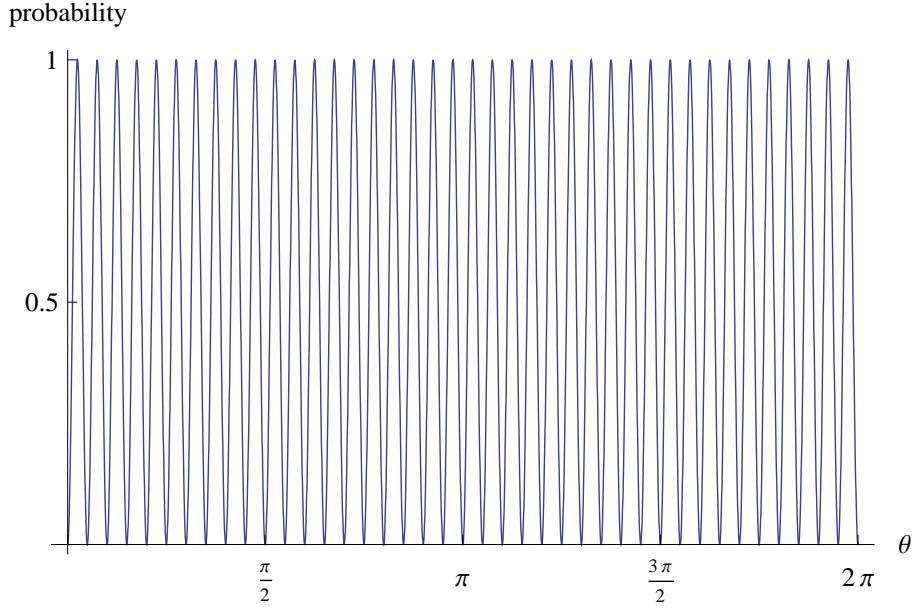}}
\end{center}
\caption{Phase estimation with noon state $n=20$}
\Label{f3}
\end{figure}%

For this purpose,
we focus on an $\epsilon$-neighborhood $U_{\theta,\epsilon}$ of $\theta$
and define 
the local asymptotic mini-max risk\cite{Hajek}:
\begin{align*}
C_\alpha [\Lambda_{\theta_0},\{(\rho_n,M^n)\}]
& :=
\lim_{\epsilon \to 0}
\limsup_{n \to\infty} 
\sup_{\theta \in U_{\theta_0,\epsilon}}
n^\alpha 
\MSE_\theta(\rho_n,M^n)
\end{align*}
and the {\it local asymptotic mini-max} bound:
\begin{align*}
C_\alpha [\Lambda_{\theta_0}]
& :=
\inf_{\{(\rho_n,M^n)\}} 
C_\alpha [\Lambda_{\theta_0},\{(\rho_n,M^n)\}].
\end{align*}

Concerning the local asymptotic mini-max bound,
we have the following two propositions.
\begin{prop}
When $\tilde{C}_\alpha [\Lambda_{\theta_0}]$ is continuous
and the convergence (\ref{3-20-1}) is compactly uniform,
\begin{align}
C_\alpha [\Lambda_{\theta_0}]
\ge 
\tilde{C}_\alpha [\Lambda_{\theta_0}].
\Label{22-9}
\end{align}
\end{prop}

\begin{proof}
For any $\delta>0$, we choose 
an integer $N$ and an $\epsilon$-neighborhood $U_{\theta_0,\epsilon}$
satisfying 
\begin{align}
C_\alpha [\Lambda_{\theta_0}] +\delta 
\ge
n^\alpha \MSE_\theta(\rho_n,M^n) ,\quad
\forall n \ge N, \forall \theta \in U_{\theta_0,\epsilon}.
\label{3-15-1}
\end{align}
We introduce two quantities
\begin{align*}
\eta_n(\theta)
& :=\int \hat{\theta} \Tr \Lambda_{\theta}(\rho_n) M^n(d \hat{\theta})\\
v_n(\theta)
& :=\int (\hat{\theta} - \eta_n(\theta))^2
\Tr \Lambda_{\theta}(\rho_n) M^n(d \hat{\theta}).
\end{align*}
Then, we obtain
\begin{align*}
\MSE_\theta(\rho_n,M^n)
=v_n(\theta)+ (\eta_n(\theta)-\theta)^2.
\end{align*}
Let $J_{\theta,n}$ be the Fisher information 
of the distribution family 
$\{\Tr \Lambda_{\theta}(\rho_n) M^n(d \hat{\theta}) |\theta \in \Theta\}$.
This quantity is smaller than $J^S[\Lambda_\theta^{\otimes n}]$.
Deforming the classical Cram\'{e}r-Rao inequality,
we obtain 
\begin{align*}
v_n(\theta)\ge 
\frac{(\frac{d \eta_n(\theta)}{d \theta})^2}{J_{\theta,n}}.
\end{align*}
Thus, we obtain
\begin{align}
\MSE_\theta(\rho_n,M^n)
\ge 
\frac{(\frac{d \eta_n(\theta)}{d \theta})^2}{J_{\theta,n}}
+ (\eta_n(\theta)-\theta)^2.\Label{22-7}
\end{align}

As is shown later, for any $\delta>0$,
there exists a sufficiently large integer $N$ satisfying the following.
For any $n \ge N$, there 
exists $\theta_n \in U_{\theta_0,\epsilon}$ such that
\begin{align}
\frac{d \eta_n(\theta_n)}{d \theta}
\ge 1-\delta.\Label{3-15-2}
\end{align}
Using (\ref{22-7}), we obtain
\begin{align}
\MSE_{\theta_n}(\rho_n,M^n) 
\ge 
\frac{ (\frac{d \eta_n(\theta)}{d \theta})^2 }{J_{\theta_n,n}}.
\Label{4-30-1}
\end{align}
Take the limit $n \to \infty$. 
Then, the continuity of $\tilde{C}_\alpha [\Lambda_{\theta_0}]$,
the compact uniformity of the convergence (\ref{3-20-1}),
(\ref{3-15-1}), and (\ref{4-30-1})
imply that
\begin{align*}
C_\alpha [\Lambda_{\theta_0}]+\delta
\ge 
(1-\delta)^2
\tilde{C}_\alpha [\Lambda_{\theta_0}].
\end{align*}
Taking the limit $\delta\to 0$, we obtain (\ref{22-9}).

Finally, we show the existence of the integer $N$ satisfying 
the above condition given by (\ref{3-15-2}) by using reduction to absurdity.
Assume that for any $\delta>0$,
there exists a subsequence $n_k$ such that
$\frac{d \eta_{n_k}(\theta)}{d \theta}< 1-\delta$
for any $\theta \in U_{\theta_0,\epsilon}$.
Thus,
\begin{align*}
& [\eta_{n_k}(\theta+\epsilon/2)-(\theta+\epsilon/2)]
-
[\eta_{n_k}(\theta-\epsilon/2)-(\theta-\epsilon/2)] \\
= &
\eta_{n_k}(\theta+\epsilon/2)-\eta_{n_k}(\theta-\epsilon/2)
-\epsilon
<- \epsilon \delta.
\end{align*}
Then,
\begin{align*}
\max \{|\eta_{n_k}(\theta+\epsilon/2)-(\theta+\epsilon/2)|,
|\eta_{n_k}(\theta-\epsilon/2)-(\theta-\epsilon/2)|\}
>
\frac{\epsilon \delta }{2}.
\end{align*}
That is,
\begin{align*}
\max \{(\eta_{n_k}(\theta+\epsilon/2)-(\theta+\epsilon/2))^2,
(\eta_{n_k}(\theta-\epsilon/2)-(\theta-\epsilon/2))^2\}
>
\frac{\epsilon^2\delta^2}{4}.
\end{align*}
Using (\ref{22-7}), we obtain
\begin{align*}
\max \{
\MSE_{\theta+\epsilon/2}(\rho_{n_k},M^{n_k}),
\MSE_{\theta-\epsilon/2}(\rho_{n_k},M^{n_k}) \}
\ge 
\frac{\epsilon^2\delta^2}{4}.
\end{align*}
Since $\MSE_{\theta+\epsilon/2}(\rho_n,M^n)$
behaves as $O(\frac{1}{n^{\alpha}})$,
we obtain contradiction.
\end{proof}

\begin{prop}
Assume that $E:=\sup_{\theta\in \Theta}|\theta|<\infty $.
When the order parameter $\alpha$ equals $1$
and $\tilde{C}_\alpha [\Lambda_{\theta_0}]$
is continuous,
\begin{align*}
C_1 [\Lambda_{\theta_0}]
= \tilde{C}_1 [\Lambda_{\theta_0}]
=\lim_{n \to \infty} 
\frac{n}{J^S[\Lambda_\theta^{\otimes n}]}.
\end{align*}
\end{prop}

\begin{proof}
As is mentioned in Section \ref{s3}, (\ref{4-23-2}) guarantees the convergence of
$\lim_{n \to \infty} 
\frac{n}{J^S[\Lambda_\theta^{\otimes n}]}$.
It is enough to show the inequality $\tilde{C}_1 [\Lambda_{\theta_0}]\ge 
C_1 [\Lambda_{\theta_0}]$.

For an arbitrary real number $\delta>0$ and an arbitrary integer $m$, 
let $\{\rho_m,M^m\}$ be a locally unbiased estimator at $\theta_0$ such that 
\begin{align}
\frac{1}{J^S[\Lambda_{\theta_0}^{\otimes m}]}+\delta
>
\MSE_{\theta_0}(\rho_m,M^m).
\Label{3-15-3-2}
\end{align}
We define another coordinate $\eta(\theta)$ by
\begin{align*}
\eta(\theta):=\int \hat{\theta} \Tr \Lambda_{\theta}(\rho_n) M^n(d \hat{\theta}),
\end{align*}
and denote the MSE concerning the parameter $\eta$ 
of an estimator $(\rho_n,M^n)$ by $\MSE_{\eta(\theta)}(\rho_n,M^n)$.
For any $\delta>0$,
we choose an integer $m$,
a sufficiently small $\epsilon'$-neighborhood $U_{\theta_0,\epsilon}$
such that
\begin{align}
\frac{\tilde{C}_1 [\Lambda_{\theta_0}]}{m} +\delta
&\ge
\frac{1}{J^S[\Lambda_{\theta_0}^{\otimes m}]}
\Label{3-15-3}
\\
1-\delta \le \frac{\eta(\theta)-\eta(\theta')}{\theta-\theta'}
& \le 1+\delta
 \Label{22-5} 
\end{align}
for $\forall \theta ,\theta' \in U_{\theta_0,\epsilon}$.

Next, let $(\rho_{nm}',{M^{nm}}' )$ be 
the estimator given as the average value of 
$n$ times applications of the estimator $(\rho_m,M^m )$ 
concerning the original parameter $\theta$.
Then, we can choose a sufficiently large number $n$ satisfying the following:
When the true parameter is $\theta_0$,
the estimate of $(\rho_{nm}',{M^{nm}}' )$ belongs to $U_{\theta_0,\epsilon}$
with the probability higher than $1-\delta$.
The second inequality in (\ref{22-5}) guarantees that
\begin{align}
& \frac{1}{n}\MSE_{\theta}(\rho_{m},{M^{m}})+\delta E \nonumber \\
=&
\MSE_{\theta}(\rho_{nm}',{M^{nm}}')+\delta E
\ge 
\frac{1}{(1+\delta)^2}
\MSE_{\eta(\theta)}(\rho_{nm}',{M^{nm}}').
\Label{22-7-a-b}
\end{align}
Thus, (\ref{3-15-3-2}), (\ref{3-15-3}), and (\ref{22-7-a-b})
imply that
\begin{align*}
(1+\delta)^2
(\frac{\tilde{C}_1 [\Lambda_{\theta_0}]}{nm} +\frac{2\delta}{n} +\delta E)
\ge
\MSE_{\eta(\theta_0)}(\rho_{nm}',{M^{nm}}').
\end{align*}
Since 
$\MSE_{\eta(\theta_0)}(\rho_{nm}',{M^{nm}}')$ is continuous
and
$(1+\delta)^2
(\frac{\tilde{C}_1 [\Lambda_{\theta_0}]}{nm} + \delta (E+2))
>
(1+\delta)^2
(\frac{\tilde{C}_1 [\Lambda_{\theta_0}]}{nm} +\frac{2 \delta}{n} + \delta E)$,
we can choose a sufficiently small number $0<\epsilon'<\epsilon$
such that
\begin{align}
& (1+\delta)^2
(\frac{\tilde{C}_1 [\Lambda_{\theta_0}]}{nm} + \delta (E+2))
\ge 
\MSE_{\eta(\theta)}(\rho_{nm}',{M^{nm}}')\Label{4-23-4}
\end{align}
for $\theta\in U_{\theta_0,\epsilon'}$.

Let $(\rho_{lnm}'',{M^{lnm}}'' )$ be 
the estimator given as the average value of 
$l$ times applications of the estimator $(\rho_{nm}',{M^{nm}}' )$ concerning $\eta$.
Since the estimator $(\rho_{nm}',{M^{nm}}' )$ 
is unbiased concerning the parameter $\eta$,
\begin{align}
\MSE_{\eta(\theta)}(\rho_{lnm}'',{M^{lnm}}'')
= \frac{\MSE_{\eta(\theta)}(\rho_{nm}',{M^{nm}}')}{l}.
\Label{22-6}
\end{align}
We choose a sufficiently large number $l$ satisfying the following:
When the true parameter is $\theta \in U_{\theta_0,\epsilon'/2}$,
the estimate $\eta$ of $(\rho_{lnm}'',{M^{lnm}}'' )$ belongs to $U_{\theta_0,\epsilon'}$
with the probability $1-p_l$,
where the probability $p_l$ exponentially goes to $0$ as $l\to \infty$.
The first inequality in (\ref{22-5})  guarantees that
\begin{align}
\frac{1}{(1-\delta)^2}
\MSE_{\eta(\theta)}(\rho_{lnm}'',{M^{lnm}}'')
+ E p_l
\ge
\MSE_{\theta}(\rho_{lnm}'',{M^{lnm}}'')
,\quad \forall
\theta \in U_{\theta_0,\epsilon'/2}.
\Label{22-7-a}
\end{align}
Thus, 
the relations (\ref{4-23-4}), (\ref{22-6}), and (\ref{22-7-a}) imply
\begin{align}
\frac{(1+\delta)^2}{l(1-\delta)^2}
(\frac{\tilde{C}_1 [\Lambda_{\theta_0}]}{nm} + \delta (E+2))
+ E p_l
\ge
\MSE_{\theta}(\rho_{lnm}'',{M^{lnm}}'')
,\quad \forall
\theta \in U_{\theta_0,\epsilon'/2}.
\Label{22-7-e}
\end{align}
Taking the limit $l \to \infty$, 
we obtain
\begin{align*}
\frac{(1+\delta)^2}{(1-\delta)^2}
(\tilde{C}_1 [\Lambda_{\theta_0}]+ nm \delta (E+2))
\ge
\lim_{l \to \infty} lnm \MSE_{\theta}(\rho_{lnm}'',{M^{lnm}}'')
,\quad \forall
\theta \in U_{\theta_0,\epsilon'/2}.
\end{align*}
Finally, using the above $n$ and $m$, 
we define a sequence of estimators $\{(\rho_{k}''',{M^{k}}''')\}$
by the following way.
For given $k$, we choose maximum $l$ such that $lnm \le k$.
Then, the estimator $(\rho_{k}''',{M^{k}}''')$ defined as
$(\rho_{lnm}'',{M^{lnm}}'')$.
In this definition, we only use $lnm$ applications, and 
the remaining $k-lnm$ applications are discarded.
So,
we obtain
\begin{align*}
\frac{(1+\delta)^2}{(1-\delta)^2}
(\tilde{C}_1 [\Lambda_{\theta_0}]+ nm \delta (E+2))
\ge
\limsup_{k \to \infty} k \MSE_{\theta}(\rho_{k}''',{M^{k}}''')
,\quad \forall
\theta \in U_{\theta_0,\epsilon'/2},
\end{align*}
which implies that
\begin{align*}
\frac{(1+\delta)^2}{(1-\delta)^2}
(\tilde{C}_1 [\Lambda_{\theta_0}]+ nm \delta (E+2))
\ge
{C}_1 [\Lambda_{\theta_0},\{(\rho_{k}''',{M^{k}}''')\}_{k=1}^{\infty}]
\ge
{C}_1 [\Lambda_{\theta_0}].
\end{align*}
Since $\delta>0$ is arbitrary,
$\tilde{C}_1 [\Lambda_{\theta_0}]\ge 
C_1 [\Lambda_{\theta_0}]$.
\end{proof}

Now, remember that the Cram\'{e}r-Rao bound 
can be attained by using the two-step method
in the case of state estimation.
By using the two-step method\cite{4,5}, 
the local asymptotic mini-max bound $C_\alpha [\Lambda_{\theta_0}]$
can be attained at all points $\theta$ as follows.
\begin{prop}\Label{22-1p}
Assume that $E:=\sup_{\theta\in \Theta}|\theta|<\infty $
and $C_1 [\Lambda_{\theta_0}]$ is continuous.
For any $\delta>0$,
there exists a sequence of estimators
$\{(\rho_n,M^n )\}$ such that
\begin{align}
C_\alpha [\Lambda_{\theta},\{(\rho_n,M^n)\}]
\le
C_\alpha [\Lambda_{\theta}]+\delta
\Label{7-24-3}
\end{align}
for all points $\theta$.
Further, when the parameter space $\Theta$ is compact,
\begin{align}
\lim_{n \to \infty}
n^{\alpha}
\min_{(\rho_n,M^n)}
\max_{\theta \in \Theta} 
\MSE_{\theta}(\rho_n,M^n )
=
\max_{\theta \in \Theta} C_\alpha[\Lambda_{\theta}].
\Label{7-24-2}
\end{align}
\end{prop}

\begin{proof}
We use a two-step method slightly different from \cite{5}.
Before applying the unknown channel $\Lambda_\theta$,
for any real number $\delta>0$,
we choose an $\epsilon_i$-neighborhood $U_{\theta_i,\epsilon_i}$
satisfying the following three conditions:
(1) $\cup_{i} U_{\theta_i,\epsilon_i}=\Theta$.
(2) For any $\theta_i$, there exists a sequence of estimators
$\{(\rho_n(\theta_i),M^n(\theta_i) )\}$ such that
$C_\alpha [\Lambda_{\theta_i}]+\delta/2
\ge \lim_{\epsilon \to 0}
\limsup_{n \to\infty} n^\alpha 
\sup_{\theta \in U_{\theta_i,\epsilon}}
\MSE_\theta(\rho_n(\theta_i) ,M^n(\theta_i) )$.
(3)
$\sup_{i: \theta \in U_{\theta_i,\epsilon_i}}
C_\alpha [\Lambda_{\theta_i}] \le
C_\alpha [\Lambda_{\theta}]+\delta/2$.

We divide $n$ applications of the unknown channel $\Lambda_\theta$
to two groups:
The first group consists of $\sqrt{n}$ applications and the second group consists of
$n-\sqrt{n}$ applications.
In the first step, we apply a POVM $M$ to the first group.
This POVM $M$ is a POVM on the single system ${\cal H}$ and 
is required to satisfy that $J_\theta$ is non-degenerate at all points $\theta$.
Based on $\sqrt{n}$ obtained data,
we estimate which $\epsilon_i$-neighborhood $U_{\theta_i,\epsilon_i}$
contains the true parameter,
and obtain the first step estimate $\theta_{\hat{i}}$.
The error probability $P_{\theta,n}$ of this step goes to $0$ exponentially,
i.e., $P_{\theta,n}$ behaves as $e^{-c \sqrt{n}}$, where $c$ depends on $\theta$.

In the second step, 
we apply the estimator $(\rho_{n-\sqrt{n}}(\theta_i) ,M^{n-\sqrt{n}}(\theta_i) )$
to the second group,
and obtain our final estimate from the outcome of 
the estimator 
$(\rho_{n-\sqrt{n}}(\theta_i),M^{n-\sqrt{n}}(\theta_i) )$.
We express this estimator by
$(\rho_n,M^n )$.
Its MSE is evaluated as
\begin{align}
\MSE_{\theta}(\rho_n,M^n )
\le
E P_{\theta,n}
+
(1-P_{\theta,n})
\sup_{i: \theta \in U_{\theta_i,\epsilon_i}}
\MSE_\theta(\rho_{n-\sqrt{n}}(\theta_i) ,M^{n-\sqrt{n}}(\theta_i) ).
\Label{7-24-1}
\end{align}
Since $n^{\alpha}E P_{\theta,n}$ goes to $0$,
we obtain
\begin{align*}
\limsup_{n \to \infty}n^{\alpha}\MSE_{\theta}(\rho_n,M^n )
\le
\sup_{i: \theta \in U_{\theta_i,\epsilon_i}}
C_\alpha [\Lambda_{\theta_i}]+\delta/2
\le
C_\alpha [\Lambda_{\theta}]+\delta.
\end{align*}
Thus, we obtain (\ref{7-24-3}).
Further, the relation (\ref{7-24-1}) yields that
\begin{align*}
& \sup_{\theta\in \Theta} \MSE_{\theta}(\rho_n,M^n ) \\
\le &
\sup_{\theta\in \Theta} E P_{\theta,n}
+
\sup_{\theta\in \Theta} (1-P_{\theta,n})
\sup_{i: \theta \in U_{\theta_i,\epsilon_i}}
\MSE_\theta(\rho_{n-\sqrt{n}}(\theta_i) ,M^{n-\sqrt{n}}(\theta_i) ).
\end{align*}
The compactness of $\Theta$ guarantees that
$\sup_{\theta\in \Theta} n^{\alpha} E P_{\theta,n} \to 0$.
Thus,
\begin{align*}
\limsup_{n \to \infty}n^{\alpha} \sup_{\theta\in \Theta} \MSE_{\theta}(\rho_n,M^n )
\le
\sup_{\theta\in \Theta} \sup_{i: \theta \in U_{\theta_i,\epsilon_i}}
C_\alpha [\Lambda_{\theta_i}]+\delta/2
\le
\sup_{\theta\in \Theta} C_\alpha [\Lambda_{\theta}]+\delta.
\end{align*}
Since the part $\ge$ of (\ref{7-24-2}) is trivial,
we obtain (\ref{7-24-2}).
\end{proof}

Proposition \ref{22-1p} holds 
even when we replace the MSE by a general error function $R(\theta,\hat{\theta})$
for one-parametric family satisfying the following conditions:
(1) 
the relation $R(\theta,\hat{\theta})\cong (\hat{\theta}-\theta)^2$ 
holds with a local coordinate
when $\hat{\theta}$ is close to $\theta$.
(2) the maximum of $R(\theta,\hat{\theta})$ exists.
Therefore, we can apply Proposition \ref{22-1p} to 
the following case:
Assume that the one-parameter channel family $\{\Lambda_\theta\}$ has a compact group covariant structure,
that is,
its parameter space is given as an interval $[a,b)$ and
there is a unitary representation $U_\theta$ of $\R$
such that
$U_{\theta'}\Lambda_\theta(\rho) U_{\theta'}^\dagger=
\Lambda_{\theta+\theta'}(\rho)$.
The error is given by $\min_{k \in \integer} (\hat{\theta}+k(b-a)-\theta )^2 $
instead of the square error $(\hat{\theta}-\theta)^2 $.
In this case, due to the group covariance,
$C_\alpha [\Lambda_{\theta}]$ does not depend on the true parameter $\theta$.
Application of Proposition \ref{22-1p} implies
that the global min-max error behaves $C_\alpha [\Lambda_{\theta}]\frac{1}{n^{\alpha}}$.

In the phase estimation case, the unknown parameter $\theta$ belongs to $[0,2\pi)$, and
the minimum of the worst value of the average error $\max_{\theta} \MSE_\theta(\rho_n,M^n)$
behaves as $\frac{\pi^2}{n^2}$ \cite{Luis,BDR,IH}.
That is, the leading decreasing order is $O(1/n^2)$ and 
the leading decreasing coefficient is $\pi^2$ when we apply the optimal estimator.
Proposition \ref{22-1p} implies that
$\max_{\theta} C_2 [\Lambda_{\theta}]=\pi^2$.
Since $C_2 [\Lambda_{\theta}]$ does not depend on $\theta$ due to the homogenous structure,
we can conclude that $C_2 [\Lambda_{\theta}]=\pi^2$.
So, the equation $J^S[\Lambda_\theta^{\otimes n}]= n^2$ implies
the equation $\tilde{C}_2 [\Lambda_{\theta}]=1$.
Hence, the Cram\'{e}r-Rao bound $\tilde{C}_2 [\Lambda_{\theta}]$ cannot be attained globally in this model.
However, it can be attained in a specific point in the following sense.

\begin{prop}\Label{pro1}
Assume that $E:=\sup_{\theta\in \Theta}|\theta|<\infty $
and $C_1 [\Lambda_{\theta_0}]$ is continuous.
For any $\delta>0$ and 
any $\theta_0\in \Theta$,
there exists a sequence of estimators
$\{(\rho_{n,\theta_0},M^n_{\theta_0} )\}$ satisfying the 
asymptotically locally unbiased condition
and the relations:
\begin{align*}
\limsup_{n \to \infty}
n^{\alpha}\MSE_{\theta}(\rho_{n,\theta_0},M^n_{\theta_0} )
& \le
C_\alpha [\Lambda_{\theta}]+\delta,
\quad\forall \theta\neq \theta_0 \\
\limsup_{n \to \infty}
n^{\alpha}\MSE_{\theta_0}(\rho_{n,\theta_0},M^n_{\theta_0} )
& \le
\tilde{C}_\alpha [\Lambda_{\theta_0}]+\delta.
\end{align*}
\end{prop}

In estimation of probability distribution,
there exists a superefficient estimator that has smaller error 
at a discrete set than the Cram\'{e}r-Rao bound\cite{L}.
Since such a superefficient estimator cannot be useful,
many statisticians think that 
it is better to impose a condition for our estimators for removing superefficient estimators.
In this classical case, if we assume the asymptotic locally unbiased condition,
we have no superefficient estimator. 
Proposition \ref{pro1} means that
even if the asymptotic locally unbiased condition is assumed,
there exists an estimator that behaves in the similar way to
a superefficient estimator in the case of unitary estimation.
So, we call such an estimator a q-channel-superefficient estimator.
That is, 
a sequence of estimators $\{(\rho_n,M^n)\}$
is called {\it q-channel-superefficient} at $\theta$ with the order $\frac{1}{n^{\alpha}}$
when $\limsup_{n \to\infty} n^\alpha 
\MSE_\theta(\rho_n,M^n)
< C_\alpha [\Lambda_{\theta}]$.
Hence, in order to remove the q-channel-superefficiency problem,
it is better to adopt the bound
$C_\alpha [\Lambda_{\theta}]$
as the criterion instead of
$\tilde{C}_\alpha [\Lambda_{\theta}]$.

\begin{proof}
We choose $\epsilon_i$-neighborhoods $U_{\theta_i,\epsilon_i}$
in the same way,
and define the neighborhood $U_{\theta_0,\frac{1}{n^{1/4}}}$.
We apply the same first step as Proposition \ref{22-1p}
to neighborhoods 
$\{U_{\theta_i,\epsilon_i}\}_i\cup\{U_{\theta_0,\frac{1}{n^{1/4}}}\}$,
and obtain the first step estimate $\theta_{\hat{i}}$.
When the first step estimate $\theta_{\hat{i}}$ is not $\theta_0$,
we apply the same method as Proposition \ref{22-1p} in the second step.
When the first step estimate $\theta_{\hat{i}}$ is $\theta_0$,
we apply the asymptotically locally unbiased estimator
whose MSE behaves as
$(\tilde{C}_\alpha [\Lambda_{\theta}]+\delta)/n^\alpha$ asymptotically.
\end{proof}

Further, since there exists an asymptotically locally unbiased estimator that surpasses 
the bound $C_\alpha [\Lambda_{\theta}]$,
the asymptotically locally unbiased condition is too weak for deriving 
the local asymptotic mini-max bound, which is more meaningful.
In order to avoid this problem, 
it is sufficient to impose the following condition:
\begin{description}
\item[(CU)]
The limit
$\lim_{n \to\infty} n^\alpha \MSE_\theta(\rho_n,M^n)$
exists for all $\theta$ 
and this convergence is compactly uniform concerning $\theta$.
\end{description}
Under the condition (CU),
$\lim_{n \to\infty} n^\alpha \MSE_\theta(\rho_n,M^n)$
is continuous concerning $\theta$, and 
\begin{align*}
\lim_{n \to\infty} n^\alpha 
\sup_{\theta \in U_{\theta_0,\epsilon}}
\MSE_\theta(\rho_n,M^n)
=
\sup_{\theta \in U_{\theta_0,\epsilon}}
\lim_{n \to\infty} n^\alpha 
\MSE_\theta(\rho_n,M^n).
\end{align*}
Thus,
\begin{align*}
& \lim_{n \to\infty} n^\alpha 
\MSE_{\theta_0}(\rho_n,M^n)
=
\lim_{\epsilon \to 0}
\sup_{\theta \in U_{\theta_0,\epsilon}}
\lim_{n \to\infty}  
n^\alpha
\MSE_\theta(\rho_n,M^n) \\
=&
\lim_{\epsilon \to 0}
\lim_{n \to\infty} n^\alpha 
\sup_{\theta \in U_{\theta_0,\epsilon}}
\MSE_\theta(\rho_n,M^n)
=
C_\alpha [\Lambda_{\theta_0},\{(\rho_n,M^n)\}].
\end{align*}
Therefore, we obtain the following corollary.
\begin{cor}
When a sequence of estimators $\{ (\rho_n,M^n)\}$
satisfies the condition (CU),
\begin{align*}
\lim_{n \to\infty} n^\alpha 
\MSE_\theta(\rho_n,M^n)
\ge
C_\alpha [\Lambda_{\theta}].
\end{align*}
\end{cor}
Therefore, 
the condition (CU) is better 
in estimation of quantum channel
than the asymptotically locally unbiased condition.

Finally, we consider the relation with the adaptive method proposed by Nagaoka\cite{N2}.
In this method, we apply our POVM to each single system ${\cal H}$,
and 
we decide the $k$-th POVM based on the knowledge of previous $k-1$ outcomes.
In this case, Fujiwara \cite{F1} analyzed 
the asymptotic behavior of the MSE of this estimator.
Now, we consider the case of $nm$ applications of the unknown channel
$\Lambda_\theta$.
In this case, we divide $nm$ applications into $n$ groups consisting of $m$ applications.
When we apply the adaptive method 
mentioned in Fujiwara\cite{F1} to these groups,
the MSE of this estimator behaves as 
$\frac{1}{n J^S[\Lambda_\theta^{\otimes m}]}$,
which is close to 
$\frac{\tilde{C}_\alpha [\Lambda_{\theta}]}{n m^{\alpha}}$.
So, when $\alpha>1$, 
this method cannot realize the optimal order $O(\frac{1}{(nm)^\alpha})$.

\section{Discussion}

We have compared 
the Cram\'{e}r-Rao bound $\tilde{C}_{\alpha}[\Lambda_\theta]$ and 
the local asymptotic mini-max bound $C_{\alpha}[\Lambda_\theta]$ 
in quantum channel estimation, which contains quantum state estimation.
When the model has group covariant structure, 
the local asymptotic mini-max bound $C_{\alpha}[\Lambda_\theta]$ coincides with the limit of the global mini-max bound.
We have also shown that both bounds 
$\tilde{C}_{\alpha}[\Lambda_\theta]$ and $C_{\alpha}[\Lambda_\theta]$ 
coincide in quantum channel estimation 
when the maximum of SLD Fisher information $J^S[\Lambda_\theta^{\otimes n}]$ behaves as $O(n)$.
The case of state estimation can be regarded as a special case of this case.
That is, the conventional state estimation has no difference between both bounds.
However,
we have shown that the Cram\'{e}r-Rao bound $\tilde{C}_{\alpha}[\Lambda_\theta]$ is different from the local asymptotic mini-max bound $C_{\alpha}[\Lambda_\theta]$ in the phase estimation.
So, we can conclude that 
the local asymptotic mini-max bound $C_{\alpha}[\Lambda_\theta]$ is more meaningful
and does not necessarily coincide with the Cram\'{e}r-Rao bound $\tilde{C}_{\alpha}[\Lambda_\theta]$.

In order to clarify the asymptotic leading order of $J^S[\Lambda_\theta^{\otimes n}]$,
we have derived the condition (C) as a sufficient condition for $J^S[\Lambda_\theta^{\otimes n}]=O(n)$.
That is, the condition $\|\Tr_{{\cal K}} D[\Lambda_\theta] \rho[\Lambda_\theta]^{-1} D[\Lambda_\theta]\|
 =\infty$ is a necessary condition for square speedup.
This condition has been derived from the following two facts.
One is the supremum of the RLD Fisher information satisfies the additive property.
The other is the RLD Fisher information is an upper bound of the SLD Fisher information.
This, the supremum of the RLD Fisher information is the upper bound of 
the regularized supremum of the SLD Fisher information,
which equals the inverse of the Cram\'{e}r-Rao bound. 
However, it is an open problem to clarify whether this upper bound can be attained by the regularized SLD Fisher information.

Further, Fujiwara and Imai \cite{FI2} 
and Matsumoto \cite{keiji}
also obtained another sufficient condition.
Since the relation with their conditions is not clear, 
its clarification is an open problem.
Our condition (C) trivially contains the case when the state 
$\rho[\Lambda_\theta]$ is a full rank state on the tensor product system
while it is not so easy to derive the above full rank condition from 
Fujiwara and Imai's condition.
Further, we have also obtained another example for 
$J^S[\Lambda_\theta^{\otimes n}]=O(n^2)$ under the condition (C).
This example is a larger class than the unitary model.
So, we can expect that 
$J^S[\Lambda_\theta^{\otimes n}]$ behaves as $O(n^2)$
if the condition (C) does not hold.
This is a challenging open problem.


\section*{Acknowledgment}
The author was partially supported by a Grant-in-Aid for Scientific Research
in the Priority Area `Deepening and Expansion of Statistical 
Mechanical Informatics (DEX-SMI)', No. 18079014
and a MEXT Grant-in-Aid for Young Scientists (A) No. 20686026.
The Centre for Quantum Technologies is funded by the Singapore Ministry of Education
and the National Research Foundation as part of the Research Centres of Excellence
programme.
The author thanks Mr. Wataru Kumagai for helpful comments.
He also thanks the referees 
for helpful comments concerning this manuscript. 
In particular, 
the first referee's report was much help to improve the presentation in Section 1.

\appendix

\section*{Lemmas needed for Theorem \ref{th1}}
\begin{lem} \label{l3-8}
Any strictly positive definite matrix $A$ and any projection $P$ satisfy the inequality
\begin{align}
A^{-1} \ge (PAP)^{-1}, \label{7-23-1}
\end{align}
where $(PAP)^{-1}$ is the inverse matrix with the domain $P$.
\end{lem}

\begin{proof}
Let $R$ be the operator norm of the matrix 
$PA(I-P)$.
Lemma \ref{l3-8-2} guarantees 
the inequality
\begin{align}
PA(I-P)+(I-P)AP
\le 
\epsilon P+
\frac{R^2}{\epsilon}(I-P)
\end{align}
for any $\epsilon >0$.
Thus,
\begin{align*}
& A \le
PAP +(I-P)A(I-P)+
\epsilon P+\frac{R^2}{\epsilon}(I-P) \\
= &
P(A+\epsilon I) P 
+(I-P)(A+ \frac{R^2}{\epsilon} I)(I-P)
\end{align*}
Since the function $x \mapsto -x^{-1}$
is operator monotone,
\begin{align*}
A^{-1} \ge
(P(A+\epsilon I) P )^{-1}
+((I-P)(A+ \frac{R^2}{\epsilon} I)(I-P))^{-1}
\ge 
(P(A+\epsilon I) P )^{-1}.
\end{align*}
Taking the limit $\epsilon \to 0$,
we obtain (\ref{7-23-1}).
\end{proof}

\begin{lem} \label{l3-8-2}
Let $A$ be a positive semi-definite matrix and $P$ be a projection.
Then, the inequality
\begin{align}
PA(I-P)+(I-P)AP
\le 
\epsilon P+
\frac{R^2}{\epsilon}(I-P) \label{7-23-2}
\end{align}
holds for any $\epsilon >0$,
where
$R$ is the operator norm of the matrix $PA(I-P)$.
\end{lem}

\begin{proof}
Choose an arbitrary normalized vector $u$.
Let $t$ be $\|Pu\|^2$.
Then, 
\begin{align*}
& \langle u| PA(I-P)+(I-P)AP |u \rangle
\le 
2 \sqrt{t}\sqrt{1-t} R \\
\le &
t \epsilon + (1-t)\frac{R^2}{\epsilon}  
=
\langle u| \epsilon P+
\frac{R^2}{\epsilon}(I-P) |u \rangle,
\end{align*}
which implies (\ref{7-23-2}).
\end{proof}

\end{document}